\documentclass[letter]{IEEEtran}
\IEEEoverridecommandlockouts
\usepackage{amsmath,amsthm,amsfonts,amssymb,amsbsy,bm,paralist,color,cite}
\usepackage{graphicx,algorithmic,algorithm,relsize}
\usepackage{multicol}
\usepackage{multirow}
\usepackage{caption}
\usepackage{epstopdf}
\usepackage{textcomp}
\usepackage{xcolor}
\usepackage{tabularx}
\usepackage{subcaption}

\makeatletter
 	\newtheorem{proposition}{Proposition}
 	
\begin{document}
	\title{Deep Reinforcement Learning Based Optimization for IRS Based UAV-NOMA Downlink Networks\\
		{\footnotesize }
		\author{Shiyu Jiao, Ximing Xie and Zhiguo Ding,~\IEEEmembership{Fellow,~IEEE}\vspace{-0.9cm}
			\thanks{Shiyu Jiao, Ximing Xie and Zhiguo Ding are with School of Electrical and Electronic Engineering, The University of Manchester, M13 9PL, U.K. (e-mail: shiyu.jiao@manchester.ac.uk, ximing.xie@manchester.ac.uk and zhiguo.ding@manchester.ac.uk).}}
	}
	
\maketitle

	\begin{abstract}
		This paper investigates the application of deep deterministic policy gradient (DDPG) to intelligent reflecting surface (IRS) based unmanned aerial vehicles (UAV) assisted non-orthogonal multiple access (NOMA) downlink networks. The deployment of the UAV equipped with an IRS is important, as the UAV increases the flexibility of the IRS significantly, especially for the case of users who have no line of sight (LoS) path to the base station (BS). Therefore, the aim of this letter is to maximize the sum rate by jointly optimizing the power allocation of the BS, the phase shifting of the IRS and the horizontal position of the UAV. Because the formulated problem is not convex, the DDPG algorithm is utilized to solve it. The computer simulation results are provided to show the superior performance of the proposed DDPG based algorithm. 
	\end{abstract}
\vspace{-0.5cm}
	\section{Introduction}
	Intelligent reflecting surfaces (IRS) have been recognized as one of the promising technologies for six-generation (6G) wireless communications\cite{8766143}, since they have shown excellent features on better spectrum-, energy-, and cost-efficiency \cite{zhao2019survey}. IRS can be viewed as a low-cost antenna array consisting of a large number of programmable reflecting elements \cite{8811733}. A variety of proven techniques, such as massive multiple-input multiple-output (massive-MIMO) and cooperative communications, only focus on how the transceiver can adapt to the channel environment, while IRS has the capability to control the wireless communication propagation environment \cite{8742603}. A typical scenario to apply IRS is when the direct links from the base station (BS) to users are blocked by buildings or mountains, which means IRS can create extra propagation paths to guarantee the quality of service (QoS).    

	Inspired by the merits of non-orthogonal multiple access (NOMA) such as high spectrum efficiency \cite{7842433}, we combine NOMA with the IRS, and \cite{9079918} has illustrated the better performance of the IRS-NOMA compare to the conventional orthogonal multiple access (OMA). On the other hand, unmanned aerial vehicles (UAV), as another 6G promising technique \cite{chowdhury20206g}, can be used to enhance the flexibility of IRS-NOMA. Our prior works \cite{jiao2020joint} jointly optimized beamforming and phase shift with pre-optimized UAV position and derived the closed-form of the optimal beamforming for a 2-user IRS-UAV-NOMA downlink system. However, the time-varying multi-user scenario is closer to the real wireless communication systems, and conventional optimization methods, like convex optimization, are hard to solve these non-convex jointly optimization problems. Alternatively, artificial intelligence (AI), such as deep learning (DL) and deep reinforcement learning (DRL), has been successfully applied to a variety of wireless communication problems \cite{ding2020harvesting,8815412,cui2019multi}. Unlike DL which needs a huge number of training labels, DRL based methods allow wireless communication systems to learn by interacting with the environment. Hence, DRL is more appropriate, as training labels are very hard to obtain in real-time wireless communication systems. Furthermore, compared to other DRL approaches such as deep Q network (DQN), deep deterministic policy gradient (DDPG) is applicable to the case with the high-dimension continuous action space.
	
 	This letter investigates the application of the DRL based methods to the multi-user IRS-UAV-NOMA downlink system. The DDPG algorithm is introduced into the DRL framework to optimize the power allocation of the BS, the phase shifting of the IRS and the horizontal position of the UAV simultaneously. Computer simulation results are provided to demonstrate the superior performance of the proposed algorithm on sum rate and robustness.

	\section{System model and Problem formulation}
	\begin{figure}[t]   
		\captionsetup{font={footnotesize}}
		\centering
		\includegraphics[width=0.6\linewidth]{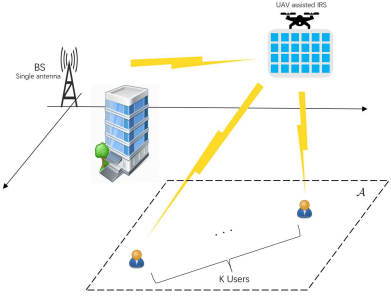}\\
		\caption{UAV based IRS-assisted NOMA downlink system.}\label{system_model}
		\vspace{-0.5cm}
	\end{figure} 
	Consider an IRS-UAV-NOMA network as shown in fig.\ref{system_model}. It is assumed that each node is equipped with a single antenna. The base station (BS) serves \textit{K} users (denote the users set by $\mathcal{K}$) who are randomly distributed in a certain area $\mathcal{A}$. Assume that their direct links to BS are blocked, for example, buildings and mountains. Hence, the UAV-equipped IRS is deployed to create reflection links between the users and the BS, where the IRS is equipped with \textit{N} passive phase shift elements. Assume that the UAV flies at a fixed altitude over area $\mathcal{A}$ autonomously, and starts at a fixed charge point. The channels are assumed as quasi-static frequency flat-fading, and the channel state information (CSI) is assumed to be known at the BS and the UAV-equipped IRS, where the energy consumption and flight duration issues of the UAV \cite{sun2013minimum} are neglected. Note that due to the used DDPG method, the proposed algorithm is applicable to the case, where the channels are time-varying between time slots, but remain constant within one time slot. Denote the channel vectors between the BS and the IRS by $\bm{g}\in \mathbb{C}^{N\times 1}$ and the channel vectors between the IRS and users by $\bm{h}_r\in \mathbb{C}^{N\times K}$, respectively. The small scale fading and the path loss are both considered. According to the NOMA principle, the BS transmits the superposition coding to all users. Hence, the received signal at each user is given by

	\begin{equation}\label{received signal}
		y_k = \bm{h}_{rk}^H\bm{\Phi} \bm{g} \sum_{i=1}^{K}\rho_i s_i + n_k, \quad k=1,\cdots, K ,
	\end{equation} 
	where $\bm{h}_{rk}\in \mathbb{C}^{N\times 1}$ is the $k$-th user's reflecting channel vector, $\Phi = \textnormal{diag}(e^{j\theta_1}, e^{j\theta_2}, \cdots, e^{j\theta_N})$ is the IRS diagonal phase shift matrix, $\theta_n \in [0, 2\pi]$ is the phase shift of the $n$-th element, $\rho_i \in [0,1]$ is the BS transmitted power allocation coefficient and $\sum_{i=1}^{K}\rho_i= 1$, $s_i$ is the transmitted signal for the $i$-th user that satisfying $\mathbb{E}[s_i^2] = 1$ and $n_k$ is the noise which follows $\mathcal{CN}(0,\sigma^2)$. 
	Since the UAV is deployed, we use $v(x,y)$ to denote the IRS-UAV horizontal position and $h_I$ for its height. The BS is located at the original point (0,0) and the BS height is $h_B$. $u_k(x_k,y_k),\quad k = 1,\cdots,K$ denotes the horizontal position of the $k$-th user. Hence, the distance between the BS and the IRS can be derived as $d_{BI}=\sqrt{x^2+y^2+(h_B-h_I)^2}$ and the distance between IRS and the $k$-th user is $d_{I{u_k}}=\sqrt{(x-x_k)^2+(y-y_k)^2+h_I^2}$. Considering the path loss, the channel gain for the $k$-th user can be rewritten as  
\vspace{-0.3cm}
	\begin{equation} \label{channelgain}
		\bm{h}_k = \frac{\bm{h}_{rk}^H \bm{\Phi} \bm{g}}{(d_{BI}d_{Iu_k})^\alpha},
	\end{equation}
	where the $\alpha$ is the path loss coefficient.
	
	To implement the successive interference cancellation (SIC) for NOMA users, the channels' quality should be gained first. Assume that and the weakest user (who has the worst channel) is the 1-st user and the strongest user (who has the best channel) is the $K$-th user. According to the SIC principle, the $j$-th ($1\leq j \leq K$) user needs to decode the signals of all $j-1$ weaker users so that the $j$-th user can remove those signal from the superposed received signal. Therefore, the signal-to-interference-plus-noise ratio (SINR) for the $j$-th user to decode the $t$-th ($t\leq j-1 \leq K$) user's signal as follows:
	\begin{equation}\label{SINRttoj}
		\textnormal{SINR}_{t\to j} = \frac{{|\bm{h}_j|}^2 P_{max}\rho_t}{\sum_{i=t+1}^{K}{|\bm{h}_j|}^2P_{max}\rho_i+\sigma^2}.
	\end{equation}
	Afterwards, the user $j$ can decode its own signal by simply treating the signal of all the rest users as interference. The SINR for the $j$-th user to decode its own signal is given by     
	\begin{equation}\label{SINRjtoj}
		\textnormal{SINR}_{j\to j} = \frac{{|\bm{h}_j|}^2 P_{max}\rho_j}{\sum_{i=j+1}^{K}{|\bm{h}_j|}^2P_{max}\rho_i+\sigma^2},
	\end{equation}
	where $P_{max}$ is the maximum transmission power. Note that the data rate for each user to decode its own signal can be calculated by (\ref{SINRjtoj}) and $R = \textnormal{log}(1+\textnormal{SINR})$. Denote the minimum target data rate by $R_{min}$. To make sure SIC can be successfully implemented, the data rate of the $j$-th user decoding the $t$-th user's signal is required no smaller than the data rate of the $t$-th user decoding its own signal, which means $R_{t \to j}\geq R_{t \to t}\geq R_{min}, \forall t<j$. The the problem formulation will be described next in detail.
	
	Our aim is to maximize the sum rate by jointly optimizing the power allocation ${\rho_i}$ at the BS, the phase-shifting $\Phi$ of the IRS and the horizontal position $v(x,y)$ of the UAV. 
	Hence, the optimization problem can be formulated as follows:
\vspace{-0.1cm}
	\begin{subequations} \label{P1}
		\begin{align}
			(P1):\max_{\{ \Theta,\rho,v\}} \quad & \sum_{t=1}^{K}R_{t\to t} \\
			\mbox{s.t.} \quad & R_{t \to t} \geq R_{min}, \forall t \in \mathcal{K}, \label{6b}\\
			& R_{t \to j} \geq R_{t \to t}   \quad \forall t, j \in \mathcal{K}, t>j, \label{6c} \\
			& \sum_{k=1}^{K}\rho_k \leq 1, \label{6d} \\
			& v(x,y)\in \mathcal{A},\label{6e} \\
			& 0 \leq \theta_n \leq 2\pi, \quad n = 1, \cdots, N. \label{6f}	
		\end{align}
	\end{subequations} 
	Constraint (\ref{6b}) is to guarantee the QoS for all users, and (\ref{6c}) ensures that the SIC processing can be implemented successfully. Constraint (\ref{6d}) is the BS total transmission power constraint and (\ref{6e}) is to restrict the UAV flies in the feasible certain area. The last constraint (\ref{6f}) is the angle constraint of each element on IRS. The problem (P1) is non-convex and hard to find the global optimal solution due to the coupled variables $\{\Theta, \rho, v\}$. Hence, in this letter, we propose a robust and low-complexity DRL based framework to solve the problem (P1).
	\vspace{-0.5cm}
	\section{Deep Deterministic Policy Gradient (DDPG)-based optimization}
	In this section, the DDPG algorithm is introduced briefly. Afterwards, we discuss how the formulated problem is processed, and how can the DDPG framework be applied. 
	
	\textit{1)  Introduction to Deep Deterministic Policy Gradient:} It is easy to know that actions of the problem (P1) are all continuous, the DQN algorithm can not be used as it can be only used for discontinuous action. The policy gradient method is unsatisfactory under wireless communication context \cite{feng2020deep}, since its drawback of convergence is very slow. Hence, we propose a DDPG-based algorithm to solve the problem (P1).
	
	DDPG is a model-free, off-policy actor-critic algorithm by applying the deep function approximators. It learns policies in high-dimension, continuous action space \cite{lillicrap2015continuous}. In DDPG algorithm, a specific action $a$ is deterministically mapped by the deterministic policy gradient (DPG) algorithm. Then critic the action by using the Q network $Q(s,a;\theta^q)$, where $\theta^q$ denotes the parameters of the critic network. Therefore, the aim of DDPG is to maximize the output Q value. Additionally, DDPG uses experience replay to solve that the hypothesis of samples is independently and identically distributed is invalid, when the samples are generated from exploring sequentially in an environment. To solve the unstable and divergence trend of Q learning, DDPG uses soft update \cite{lillicrap2015continuous}, which means updating target network's parameters by slowly tracking the learned evaluation network. For implementing soft update, we first copy the actor network $Q'(s,a;\theta^{q'})$ and the critic network $\mu'(s|\theta^{\mu'})$ as two target networks to calculate the corresponding target values, and then update their parameters according to the equation
	\begin{equation}\label{softupdate}
		\theta' \leftarrow \tau \theta + (1-\tau)\theta ,
	\end{equation} 
	where $\tau\ll 1$. On the other hand, exploration is also a challenge for learning in continuous action space \cite{lillicrap2015continuous}. In DDPG algorithm, adding randomly generated noise is a simple and efficient way to implement exploration
	\begin{equation}\label{exploration}
		\mu'(s_t) = \mu(s_t;\theta^\mu) + \mathcal{N},
	\end{equation} 
	where $\mathcal{N}$ is the noise that depends on the environment. 
	
	\textit{2) The DRL processing:} In the communication system model fig.\ref{system_model}, we define the time-varying channels as the environment and treat the IRS-UAV as the agent. The rest corresponding elements are defined as 
	
	$\bullet$ \textit{State space:} the state of time step $t$ is defined as
	\begin{equation}\label{st}
		\begin{aligned}
		s_t = & \left[R_{1}^{(t-1)},\cdots,R_{K}^{(t-1)}, \theta_1^{(t-1)},\cdots,\theta_{2N}^{(t-1)},\right.
		\\ \nonumber
		 & \left. \rho_1^{(t-1)},\cdots,\rho_k^{(t-1)},x^{(t-1)},y^{(t-1)} \right],
	 \end{aligned}
	\end{equation}
	where $\{R_{1}^{(t-1)},\cdots,R_{K}^{(t-1)}\}$ are all users' data rate at time $t-1$, $\{\theta_1^{(t-1)},\cdots,\theta_{2N}^{(t-1)}\}$ denotes real and imaginary parts of the IRS phase shift, $\{\rho_1^{(t-1)},\cdots,\rho_k^{(t-1)}\}$ denotes the power allocation to each user's signal and $\{x^{(t)},y^{(t)}\}$ represents the UAV's horizontal position. 
	
	$\bullet$ \textit{Action space:} the action of time step $t$ is defined as
	\begin{equation}\label{at}
		a_t = \begin{bmatrix}
			\theta_1^{(t)},\cdots,\theta_{2N}^{(t)}, \rho_1^{(t)},\cdots,\rho_k^{(t)},x^{(t)},y^{(t)}
		\end{bmatrix}.
	\end{equation}
	At the time step $t$, the agent inputs the state $s_t$ to obtain the corresponding action $a_t$ according to current environment. Then the agent gains the new phase shift $\Phi$, power allocation $\rho_i, i=1,\cdots,k$ and horizontal position $v$.
	
	$\bullet$ \textit{Reward:} Define the sum rate as the reward, which is consistent with our aim
	\begin{equation}\label{rt}
		r_t = R_{sum}^{(t)} = \sum_{k=1}^{K}R_k^{(t)}, k = 1, \cdots, K.
	\end{equation}

	\textit{3) Processing for the formulated problem:} To satisfy the constraints of the problem (P1), we carry out the following manipulations: To guarantee QoS constraint (\ref{6b}), calculate the data rate $R_k^{(t)}$ of each user at each step $t$ to decide if it can achieve the minimum rate. If all the calculated rates satisfy the constraint (\ref{6b}), then store this experience into the replay buffer.  
	If not, set harsh punishment (i.e. zero or a negative reward) for the bad experience then store, to avoid the agent taking bad actions that can not satisfy the constraint (\ref{6b}). If the conventional optimization methods are used, such as convex optimization, the most complicated constraint is (\ref{6c}) who is to ensure the SIC successfully implement. Nevertheless, in our proposed algorithm, channel vectors are calculated after the action $a_t$ comes out at each step $t$ to decide the decoding order, since decoding order depends on the current channel environment and it always satisfies the constraint (\ref{6c}) (see the proposition 1). Note that channel vectors (\ref{channelgain}) are still changing due to the output action of each step even the generated channel vectors are constant in each episode. 
	
	\begin{proposition}
		The SIC constraint (\ref{6c}) will always be satisfied if the decoding order is decided by current channel quality.
	\end{proposition}	
	\begin{proof}
		Recall equation (\ref{SINRttoj}), its numerator and denominator are divided by $|h_j|^2$ simultaneously (where the case for the weaker $t$-th user shown in (4) can be obtained similarly), then we have
		\begin{equation}
			\begin{aligned}
				\textnormal{SINR}_{t\to j} = \frac{P_{max}\rho_t}{\sum_{i=t+1}^{K}P_{max}\rho_i+\frac{\sigma^2}{|h_j|^2}}
			\end{aligned},
		\end{equation}
		\begin{equation}
			\begin{aligned}
				\textnormal{SINR}_{t\to t} = \frac{P_{max}\rho_t}{\sum_{i=t+1}^{K}P_{max}\rho_i+\frac{\sigma^2}{|h_t|^2}}
			\end{aligned}.
		\end{equation}
		Under the given $|h_j| \geq |h_t|$, we have $\textnormal{SINR}_{t\to j}\geq \textnormal{SINR}_{t\to t}$ that satisfies the SIC constraint.
	\end{proof} 
Therefore, the problem (P1) becomes
	\begin{subequations} \label{P2}
		\begin{align}
			(P2):\max_{\{ \Theta,a,v\}} \quad & \sum_{t=1}^{K}R_{t\to t} \\
			\mbox{s.t.} \quad & R_{t\to t}   \geq R_{min}, \quad \forall t \in \mathcal{K},\\
			\quad & (\textnormal{\ref{6d}}) - (\textnormal{\ref{6f}}).	
		\end{align}
	\end{subequations} 
	
	For the constraint (\ref{6d}), we found that the output from the neural network highly possible exists negative value. To solve this, we can use some functions (e.g. exponential function) to map the output value to the feasible range, and this trick is also valid for the constraint (\ref{6e}). The last constraint (\ref{6f}) can be easily satisfied by reforming the phase shifting part $\{\theta_1^{(t)},\cdots,\theta_{2N}^{(t)}\}$ of the output action  into the diagonal matrix $\Phi^{(t)} = \textnormal{diag}(e^{\theta_1^{(t)}},\cdots,e^{\theta_N^{(t)}})$. 
	
	\textit{4) Working Procedure:} Build four neural networks which are the target actor network $\theta_{\mu'}$, the evaluation actor network $\theta_{\mu}$, the target critic network $\theta_{q'}$ and the evaluation critic network $\theta_{q}$, where $\theta_{\mu'}$ and  $\theta_{\mu}$ use a consistent structure (the same for $\theta_{q'}$ and $\theta_{q}$ ). Creating experience replay buffer $\mathcal{B}$ with capacity $\mathcal{C}$. The whole working procedure is divided into the training data collection stage and the training stage. In each episode, randomly initialize the channel gains $h_k$, the phase shift $\Phi$ and the users' position $u$ in area $\mathcal{A}$. The UAV horizontal position $v$ is initialized at a fixed point. Power allocation $\rho$ is averagely initialized. Afterwards, calculate all users' data rates $R_i^{(t)}, i = 1,\cdots,K$ to obtain the initial state $s_t$. Then the corresponding action $a_t$ will be calculated out by the evaluation actor network by taking the current state $s_t$ as input. The reward $r_t$ is calculated by the equation (\ref{rt}) and the next state $s_{t+1}$ is gained according to (\ref{st}). Store $\{s_t,a_t,r_t,s_{t+1}\}$ into the replay buffer $\mathcal{D}$ as one transition. Once the replay buffer is full, the training stage starts to carry out. During each episode, first store the current transition $\{s_t,a_t,r_t,s_{t+1}\}$ by replacing the oldest transition, then the critic evaluation network samples $N_B$ transitions $\{s_i,a_i,r_i,s_{i+1}\} (i = 1,\cdots, N_B)$ as a minibatch whose size is $N_B$ from the replay buffer $\mathcal{D}$ to calculate the target Q value $y_i$:
	\begin{equation}\label{yi}
		y_i=\left\{
		\begin{array}{lr}
			r_i ,  \\
			r_i+\lambda Q'(s_{i+1},\mu'(s_{j+1};\theta_{\mu'};\theta_{q'})), 
		\end{array}
		\right.
	\end{equation}
	where $\lambda$ is the discount factor for future reward. Then minimizing the loss function (\ref{lf}) to update the critic evaluation network
	\begin{equation}\label{lf}
		L(\theta_q) = \frac{1}{N_B} \sum_{i=1}^{N_B}(y_i-Q(s_i,a_i;\theta_q))^2.
	\end{equation}
	After that, the actor evaluation network will be updated by applying the gradient ascent  (\ref{pg})
	\begin{equation}\label{pg}
		\nabla_{\theta_{\mu}}J =  \frac{1}{N_B} \sum_{i=1}^{N_B} (\nabla_aQ(s_i, \mu(s_i;\theta_{\mu});\theta_q)|\nabla_{\theta_{\mu}}\mu(s_i;\theta_{\mu})).
	\end{equation}
	Next, the actor target network and critic target network should be updated by the soft update method (\ref{softupdate}). The algorithm \ref{alg:1} describes in more detail.
	
	\textit{Complexity analysis:} The time complexity of DDPG is no greater than $2\mathcal{O}(D_{in}^4+D_{in}^5)$ (where $D_{in} = 2(K+N+1)$ is the input dimension) which is smaller than the complexity $\mathcal{O}((6K^2+N^2+M)^{3.5})$ of the semidefinite relaxation (SDR)-based optimization  \cite{mu2019exploiting}, which is consistent with the running time results of \cite{9322372}.
	\vspace{-0.3cm}
	\section{Numerical Results}
	In this section, we carry out the proposed DDPG-based algorithm and present the results to analyse its performance. As the fig. (\ref{system_model}) shown, the BS is deployed at the origin point (0,0), the IRS-UAV starts at the point (50,0), and users are randomly distributed in the area $\mathcal{A}$ which is (45,45), (55,45), (55,55) and (45,55). In each episode, users' positions are assumed fixed. As assumed before, the channels between the BS and the IRS, and the channels between the IRS and users are all LoS, the Rician fading channel is used according to the following equation: 
	\begin{equation}
		\bm{G} = \bm{\overline{H}}\sqrt{\frac{\Omega}{\Omega+I_N}} + \bm{H_R}\sqrt{\frac{1}{\Omega+I_M}} ,
	\end{equation}
	where $\bm{\overline{H}}$ represents the deterministic component, $\bm{H_R}$ denotes the Rayleigh fading component and $\Omega$ is the Rician K-factor. In our simulations, we set $\Omega = 10$. For the large scale fading, the path loss coefficient is $\alpha = 2$. According to the channel assumption in the section I, the channels are randomly generated for each episode, but they are fixed within each episode. On the other hand, the altitude of the BS is $h_B = 20$ and deploy the IRS-UAV at $h_U = 30$. For other parameters, we set noise power is $\sigma^2 = -60$dB and the minimum target rate for each user is $R_{min} = 1.2\textnormal{(bps/Hz)}$.
	For the Actor network, the two layers fully connected network (i.e. two-layered DNN) is used for both actor evaluation network and actor target network. The dimensions of the input layer and the output layer are $2(N+K+1)$ and $2N+K+2$ respectively which are determined by the dimensions of state and action. Note that a complex number should be treated as a two-dimension array (real part and imaginary part). On the other hand, the first layer uses the ReLU function as the activation function while the output layer uses tanh($\cdot$) function to gain enough gradient. For the critic network, similarly, we use two layers fully connected network. However, the structure becomes the following: input the state to one layer and input the action to another layer, then add these two layers' output together, then through the ReLU function as the input of the output layer. Batch normalization is used for both the actor and the critic. The setting of hyper parameters is following: learning rate for training evaluation network $\beta = 0.001$, discount factor $\lambda = 0.95$, learning rate for soft update $\tau = 0.005$, size of replay buffer $\mathcal{C} = 50000$, number of episodes $J = 1000$, number of steps in each episode $T = 500$ and size of sampled mini-batch $N_B = 16$. Additional, the added noise in equation (\ref{exploration}) for exploration is set as complex Gaussian noise with zero mean and 0.1 variance.\\

\begin{figure}
	\centering
	\begin{subfigure}{0.24\textwidth}
		\centering
		\includegraphics[width=\textwidth]{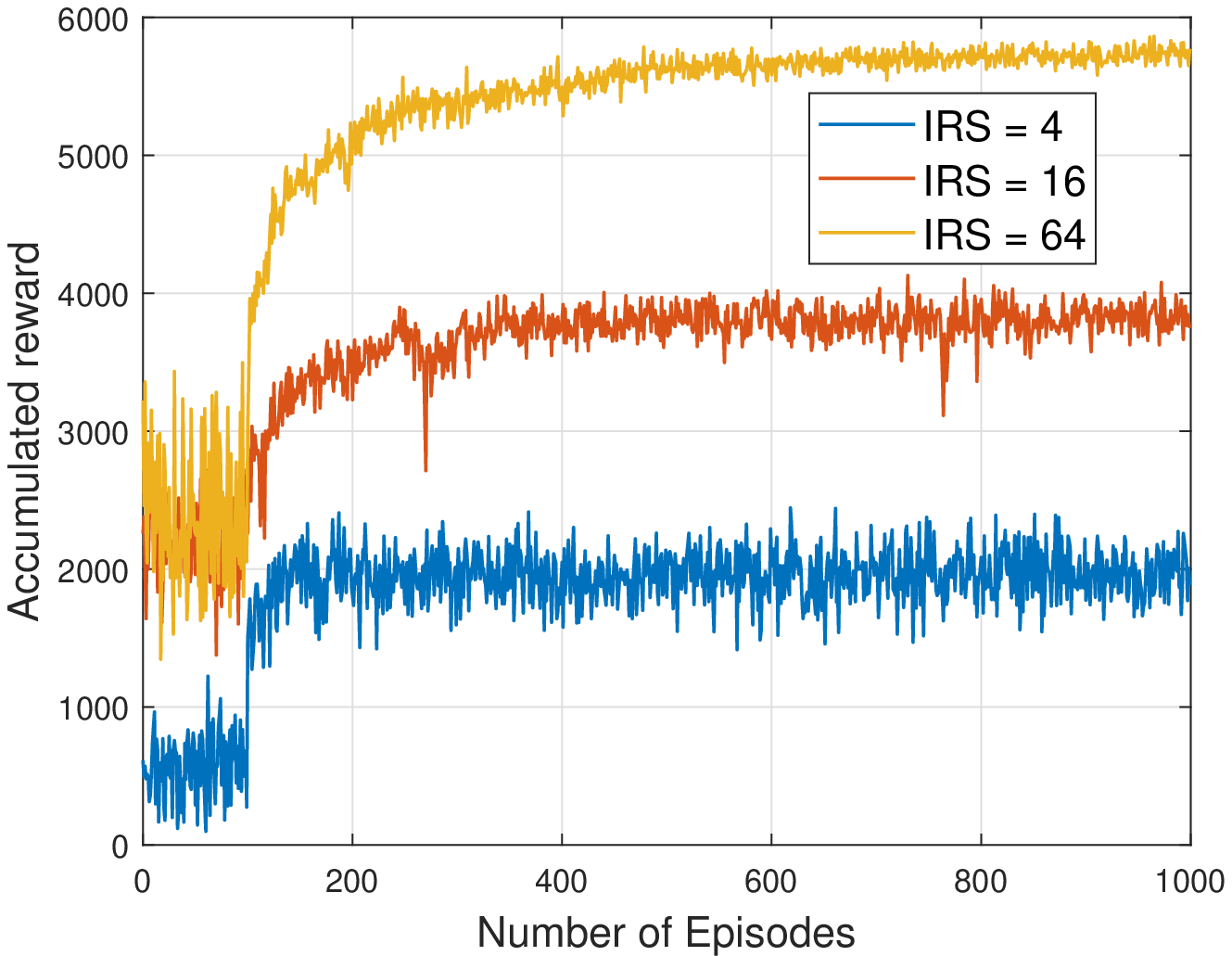}
		\captionsetup{font={scriptsize}}
		\caption{Number of episodes versus accumulated reward for different number of IRS elements $P_t = 10$dB, K =4.}
		\label{IRSvsAR}
	\end{subfigure}%
\hfill
	\begin{subfigure}{0.24\textwidth}
		\centering
		\includegraphics[width=\textwidth]{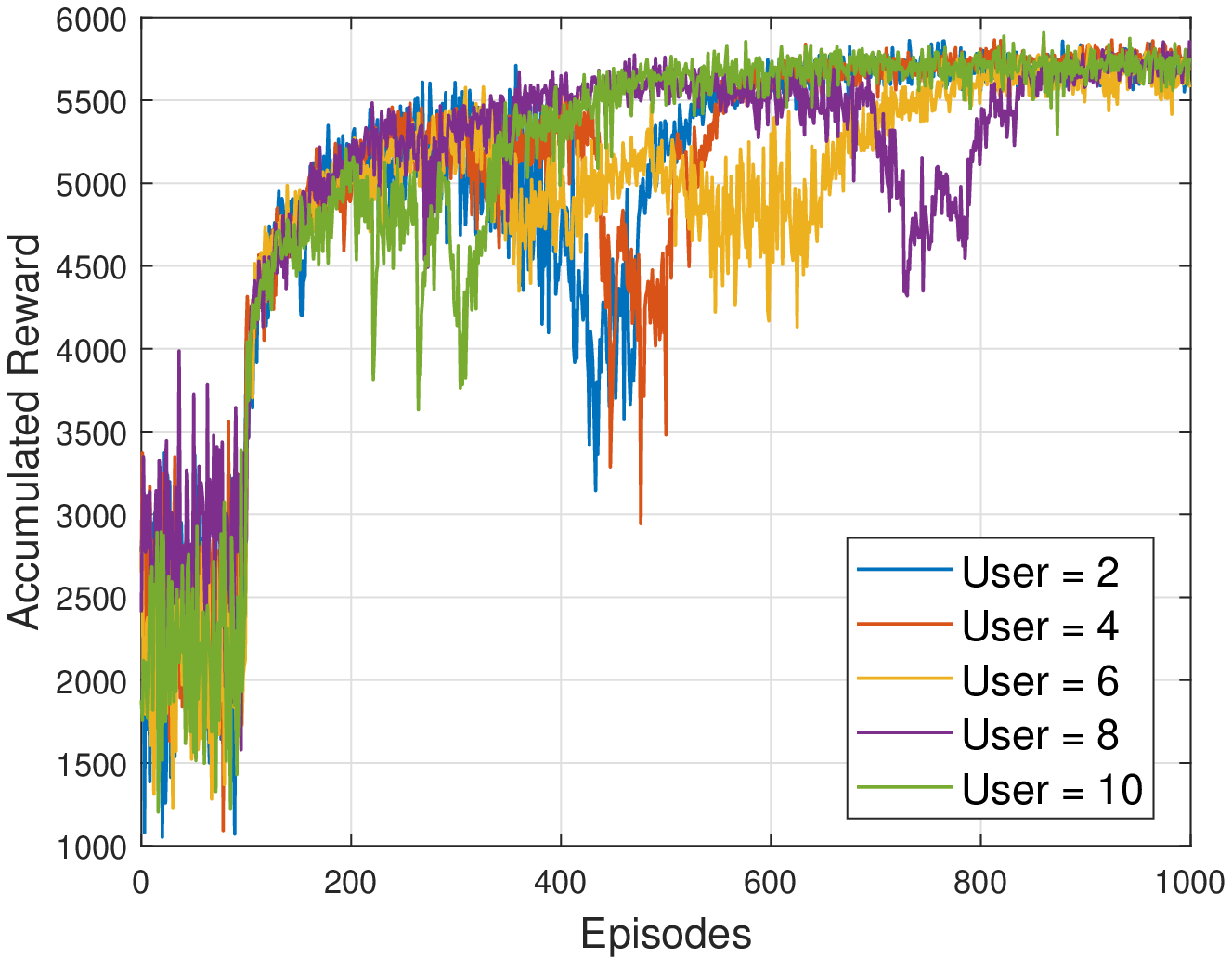}
		\captionsetup{font={scriptsize}}
		\caption{Number of episodes versus accumulated reward for different number of users $P_t = 10$dB, $N = 64$.}
		\label{UservsAR}
	\end{subfigure}
\captionsetup{font={footnotesize}}
\caption{Convergence for different scenarios. }\label{IRSUservsAR}
\vspace{-0.5cm}
\end{figure}
\vspace{-0.4cm}
	In fig.\ref{IRSUservsAR}, the number of episodes versus accumulated reward is shown under different IRS and users setup, where their first 100 episodes is the random data collection stage. The number of neurons for each hidden layer is 300. Fig.\ref{IRSvsAR} illustrates more IRS elements are used, the higher accumulated reward can be obtained. In addition, comparing these three cases, the IRS = 4 case converges before 200 episodes, the IRS = 16 case converges before 400 episodes and the IRS = 64 case converges at around 800 episodes. Hence, for the same DDPG framework training, the fewer the number of elements, the faster the convergence. Hence, increase the number of neurons can improve the convergence speed, but more neurons lead to more calculations. Therefore, it is crucial to build a neural network according to the actual situation. Fig.\ref{UservsAR} reveals what will happen when a BS serves a different number of users. It is clear that these five scenarios start at different levels at the random initialization stage, but converges at the same level after around 800 episodes. In consequence, in this system when the transmit power and the number of IRS elements are fixed, the sum rate cannot be improved by increasing the number of users, as the degrees of freedom available for resource allocation are limited in a downlink system. Hence, it is important to consider that the tradeoff problem between the number of users and the data rate when designing the system. On the other hand, no matter how many IRS elements or users there are, the proposed algorithm is convergent and stable (i.e. it is robust to the number of IRS elements and users).
	
	\begin{figure}
		\centering
		\begin{subfigure}{0.24\textwidth}
			\centering
			\includegraphics[width=\textwidth]{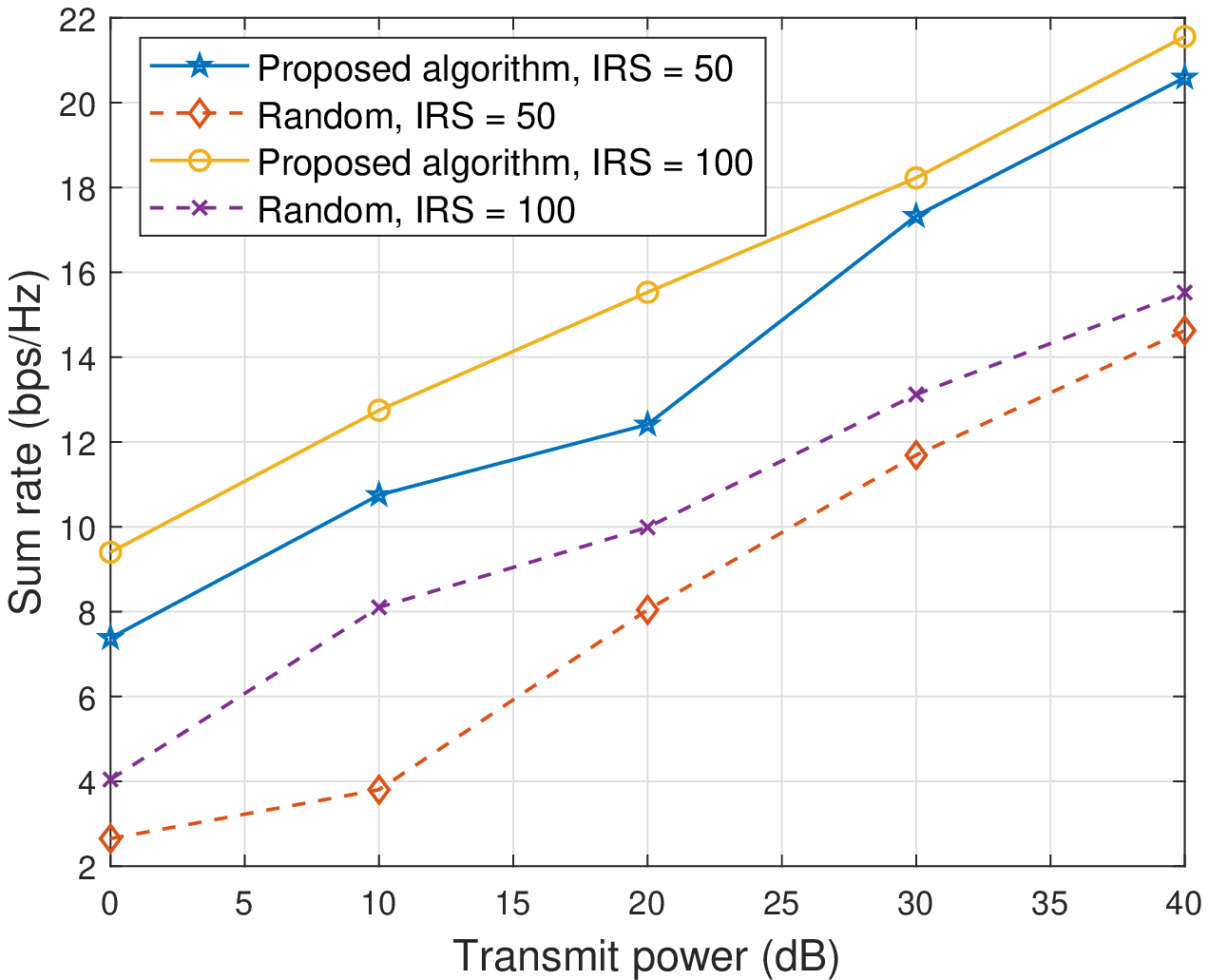}
			\captionsetup{font={scriptsize}}
			\caption{Transmit power versus sum rate $K = 4$.}
			\label{PvsR}
		\end{subfigure}%
		\hfill
		\begin{subfigure}{0.24\textwidth}
			\centering
			\includegraphics[width=\textwidth]{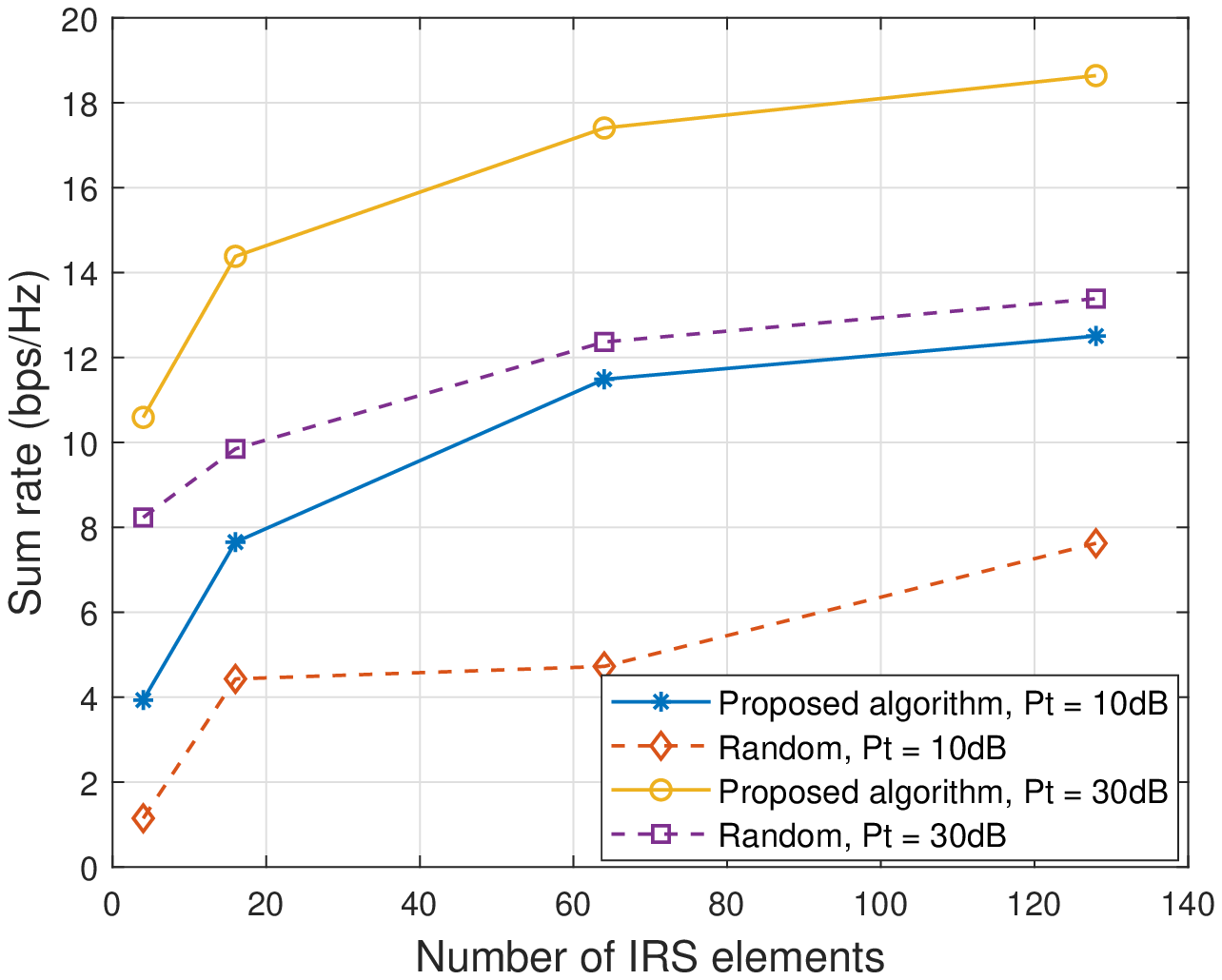}
			\captionsetup{font={scriptsize}}
			\caption{Number of IRS elements versus sumrate, $K = 4$.}
			\label{IRSvsR}
		\end{subfigure}
	\captionsetup{font={footnotesize}}
		\caption{Sum rate comparison for different scenarios.}\label{PIRSvsR} 
		\vspace{-0.5cm}
\end{figure}

	Fig.\ref{PvsR} illustrates the sum rate versus maximum transmit power $P_t$. Consider two cases of system parameters setting, one is IRS elements $N = 50$ and the other one is $N = 100$. As can be seen, the proposed algorithm outperforms the random case significantly for all considered transmit power, even the optimized case for $N = 50$ is much better than the random case for $N = 100$.  	
	To further demonstrate the proposed algorithm's performance, we carried out the algorithm for scenarios of different number of IRS elements, as Fig.\ref{IRSvsR} shown. It can be seen that the sum rate increases with the increase of IRS elements quantity. Nevertheless, the phase shift is complex, one IRS elements increasing means two more training data, and the number of neuroses depends on the training data. Although increasing IRS elements is a good way to enhance the sum rate, a huge number of IRS and neurons will cause higher calculation complexity then make non-negligible output latency. Hence, the tradeoff problem on sum rate and complexity has to be considered in practical construction.
\vspace{-0.3cm}
	\section{Conclusion}
	
	This paper investigated the sum rate maximizing problem in an IRS-UAV-NOMA downlink network. Power allocation of the BS, the IRS phase shift and the UAV position are jointly optimized by applying the proposed DDPG based algorithm efficiently. Rearranging the decoding order according to the current channel environment in each step is an efficient way to guarantee the SIC implementing successfully. Computer simulations have shown that the proposed algorithm can be applied in the time-varying channel environment to enhance the sum-rate performance significantly, as well as is robust to the number of IRS elements and users.	
	\setlength{\textfloatsep}{0pt}
		\begin{algorithm}[t]
		\caption{Proposed DDPG-based algorithm}
		\label{alg:1}
		\begin{algorithmic}[1]
			\STATE \textbf{Initialization:} Randomly initialize the critic evaluation network $Q(s,a;\theta_q)$ and the actor evaluation network $\mu(s;\theta_{\mu})$ with their corresponding parameters $\theta_q$ and $\theta_{\mu}$. Initialize the critic target network $Q'(s,a;\theta_{q'})$ and the actor target network $\mu'(s;\theta_{\mu'})$ with parameters $\theta_{q}\leftarrow \theta_{q'}$ and $\theta_{\mu}\leftarrow\theta_{\mu'}$.
			
			Initialize the experience replay buffer $\mathcal{D}$ with capacity $\mathcal{C}$.
			
			Initialize the learning rate $\beta$, the discount factor $\lambda$, the soft update coefficient $\tau$ and the minibatch size $N_B$. 
			\FOR{episode $j = 1, \cdots, J$}
			\STATE Randomly initialize the phase shift matrix $\Phi^{(j)}$ and obtain channel vectors $G^{(j)}$ and $h_{rk}^{(j)}$,  users' position $u_k(x_k,y_k)$. Initialize the UAV's position $v(x,y)$ at a fixed point. Average initializing the power allocation coefficient $\rho_k = \frac{1}{K}, k = 1, \cdots, K$.
			\STATE Decide the decoding order according to (\ref{channelgain}).
			\STATE Calculate each user's data rate by using (\ref{SINRjtoj}).
			\STATE Obtain the initial observed state $s_1$ (\ref{st}).
			\STATE Initialize the random process $\mathcal{N}$ for action exploration.
			\FOR{step $t = 1, \cdots,T$}
			\STATE Choose action $a_t = \mu(s;\theta_{\mu}) + \mathcal{N}_t$.
			\STATE Extract corresponding actions to obtain phase shift $\Phi^{(t)}$, power allocation coefficients $\rho^{(t)}$, UAV position $v$ and reform $\Phi^{(t)}$ into $\Phi^{(t)} = \textnormal{diag}(e^{\theta_1^{(t)}},\cdots, e^{\theta_N^{(t)}})$.
			\STATE Decide the current decoding order according to the current channel quality.
			\STATE Calculate each user's date rate according to (\ref{SINRjtoj}) and the sum rate to obtains the reward $r_t$ and the new state $s_{t+1}$.
			\STATE Set $r_t$ to 0 if (\ref{6b}) can not be satisfied.
			\STATE Store transition $\{s_t,a_t,r_t,s_{t+1}\}$ into the replay buffer $\mathcal{D}$.
			\STATE Sample $N_B$ minibatch transitions from $\mathcal{D}$ to train.
			\STATE Calculate target Q value by the equation (\ref{yi}).
			\STATE Update the critic evaluation network $Q(s,a;\theta_q)$ by minimizing the loss function (\ref{lf}).
			\STATE Update the actor evaluation network $\mu(s;\theta_{\mu})$ by using the sampled policy gradient in (\ref{pg}) .
			\STATE Update two target networks by using soft update (i.e.(\ref{softupdate})).
			\STATE Transfer state $s_t$ to $s_{t+1}$.
			\ENDFOR
			\ENDFOR
		\end{algorithmic}  
			
	\end{algorithm}
	\vspace{12pt}
	\bibliographystyle{IEEEtran}
	\vspace{-0.5cm}
	\bibliography{myrefsletter.bib}
\end{document}